\documentclass{article}
\usepackage{amsmath}
\usepackage{amssymb}
\usepackage{amsfonts}
\usepackage{hyperref}

\newtheorem{theorem}{Theorem}

\newtheorem{corollary}[theorem]{Corollary}

\newtheorem{lemma}[theorem]{Lemma}

\newtheorem{proposition}[theorem]{Proposition}
\newtheorem{remark}[theorem]{Remark}

\newenvironment{proof}[1][Proof]{\noindent\textbf{#1.} }{\ \rule{0.5em}{0.5em}}
\input{tcilatex}
\numberwithin{theorem}{section}
\numberwithin{equation}{section}

\begin{document}

\title{The Finsler-like geometry of the $(t,x)$-conformal deformation of the
jet Berwald-Mo\'{o}r metric}
\author{Mircea Neagu \\
%EndAName
{\small Revised March 20, 2012; One added Remark 4.5 }}
\date{}
\maketitle

\begin{abstract}
The aim of this paper is to develop on the $1$-jet space $J^{1}(\mathbb{R}%
,M^{n})$ the Finsler-like geometry (in the sense of distinguished (d-)
connection, d-torsions, d-curvatures and some gravitational-like and
electromagnetic-like geometrical models) attached to the $(t,x)$-conformal
deformation of the Berwald-Mo\'{o}r metric.
\end{abstract}

\textit{2010 Mathematics Subject Classification:} 53C60, 53C80, 83C22.

\textit{Key words and phrases:} $(t,x)$-conformal deformed jet Berwald-Mo%
\'{o}r metric, canonical nonlinear connection, Cartan canonical connection,
d-torsions and d-curvatures, geometrical Einstein-like equations.

\section{Introduction}

\hspace{5mm}The geometric-physical Berwald-Mo\'{o}r structure (\cite{Berwald}%
, \cite{Moor-1}, \cite{Moor-2}) was intensively investigated by P.K.
Rashevski (\cite{Rashevski}) and further fundamented and developed by D.G.
Pavlov, G.I. Garas'ko and S.S. Kokarev (\cite{Pavlov-1}, \cite{Pavlov-2}, 
\cite{Garasko-Pavlov}, \cite{Pavlov-Kokarev}). At the same time, the
physical studies of Asanov (\cite{Asanov[1]}) or Garas'ko and Pavlov (\cite%
{Garasko-Pavlov}) emphasize the importance of the Finsler geometry
characterized by the total equality in rights of all non-isotropic
directions in the theory of space-time structure, gravitation and
electromagnetism. For such a reason, one underlines the important role
played by the Berwald-Mo\'{o}r metric%
\begin{equation*}
F:TM\rightarrow \mathbb{R},\mathbb{\qquad }F(y)=\sqrt[n]{y^{1}y^{2}...y^{n}},%
\mathbb{\qquad }n\geq 2,
\end{equation*}%
whose tangent Finslerian geometry is studied by geometers as Matsumoto and
Shimada (\cite{Mats-Shimada}) or Balan (\cite{Balan}). In such a
perspective, according to the recent geometric-physical ideas proposed by
Garas'ko in \cite{Garasko-1} and \cite{Garasko-2}, we consider that a
Finsler-like geometric-physical study for the $(t,x)$-conformal deformations
of the jet Berwald-Mo\'{o}r structure is required. Consequently, this paper
investigates on the $1$-jet space $J^{1}(\mathbb{R},M^{n})$ the Finsler-like
geometry (together with a theoretical-geometric gravitational field-like
theory) of the $(t,x)$\textit{-conformal deformation of the Berwald-Mo\'{o}r
metric}\footnote{%
We assume that we have $y_{1}^{1}y_{1}^{2}...y_{1}^{n}>0$. This is a domain
of existence where we can $y$-differentiate the Finsler-like function $%
\overset{\ast }{F}(t,x,y)$.}%
\begin{equation}
\overset{\ast }{F}(t,x,y)=e^{\sigma (x)}\sqrt{h^{11}(t)}\cdot \left[
y_{1}^{1}y_{1}^{2}...y_{1}^{n}\right] ^{\frac{1}{n}},  \label{rheon-B-M}
\end{equation}%
where $\sigma (x)$ is a smooth non-constant function on $M^{n}$, $h^{11}(t)$
is the dual of a Riemannian metric $h_{11}(t)$ on $\mathbb{R}$, and%
\begin{equation*}
(t,x,y)=(t,x^{1},x^{2},...,x^{n},y_{1}^{1},y_{1}^{2},...,y_{1}^{n})
\end{equation*}%
are the coordinates of the $1$-jet space $J^{1}(\mathbb{R},M^{n})$, which
transform by the rules (the Einstein convention of summation is assumed
everywhere):%
\begin{equation}
\widetilde{t}=\widetilde{t}(t),\quad \widetilde{x}^{i}=\widetilde{x}%
^{i}(x^{j}),\quad \widetilde{y}_{1}^{i}=\dfrac{\partial \widetilde{x}^{i}}{%
\partial x^{j}}\dfrac{dt}{d\widetilde{t}}\cdot y_{1}^{j},  \label{tr-rules}
\end{equation}%
where $i,j=\overline{1,n},$ rank $(\partial \widetilde{x}^{i}/\partial
x^{j})=n$ and $d\widetilde{t}/dt\neq 0$. Note that the particular jet
Finsler-like geometries (together with their corresponding jet geometrical
gravitational field-like theories) of the $(t,x)$-conformal deformations of
the Berwald-Mo\'{o}r metrics of order three and four are now completely
developed in the papers \cite{Neagu-x-B-M(3)} and \cite{Neagu-x-B-M(4)}.

Based on the geometrical ideas promoted by Miron and Anastasiei in the
classical Lagrangian geometry of tangent bundles (\cite{Mir-An}), together
with those used by Asanov in the geometry of $1$-jet spaces (\cite{Asanov[2]}%
), the differential geometry (in the sense of d-connections, d-torsions,
d-curvatures, gravitational and electromagnetic geometrical theories)
produced by an arbitrary jet rheonomic Lagrangian function $L:J^{1}(\mathbb{R%
},M^{n})\rightarrow \mathbb{R}$ is now exposed in the monograph \cite%
{Balan-Neagu}. In what follows, we apply the general jet geometrical results
from book \cite{Balan-Neagu} to the $(t,x)$-conformal deformed jet Berwald-Mo%
\'{o}r metric (\ref{rheon-B-M}).

\section{The canonical nonlinear connection}

\hspace{5mm}Let us rewrite the $(t,x)$-conformal deformed jet Berwald-Mo\'{o}%
r metric (\ref{rheon-B-M}) in the form%
\begin{equation*}
\overset{\ast }{F}(t,x,y)=e^{\sigma (x)}\sqrt{h^{11}(t)}\cdot \left[
G_{1[n]}(y)\right] ^{1/n},
\end{equation*}%
where%
\begin{equation*}
G_{1[n]}(y)=y_{1}^{1}y_{1}^{2}...y_{1}^{n}.
\end{equation*}%
Hereinafter, the \textit{fundamental metrical d-tensor} produced by the
metric (\ref{rheon-B-M}) is given by the formula\footnote{%
Throughout this paper the Latin letters $i,$ $j,$ $k,$ $m,$ $r,...$ take
values in the set $\{1,2,...n\}$.} (see \cite{Balan-Neagu})%
\begin{equation*}
\overset{\ast }{g}_{ij}(t,x,y)\overset{def}{=}\frac{h_{11}(t)}{2}\frac{%
\partial ^{2}\overset{\ast }{F^{2}}}{\partial y_{1}^{i}\partial y_{1}^{j}}%
\Rightarrow
\end{equation*}%
\begin{equation}
\overset{\ast }{g}_{ij}(t,x,y):=\overset{\ast }{g}_{ij}(x,y)=\frac{%
e^{2\sigma (x)}}{n}\left( \frac{2}{n}-\delta _{ij}\right) \frac{%
G_{1[n]}^{2/n}}{y_{1}^{i}y_{1}^{j}},  \label{BM-fdt-metr-d-tensor}
\end{equation}%
where we have no sum by $i$ or $j$. Moreover, the matrix $\overset{\ast }{g}%
=(\overset{\ast }{g}_{ij})$ admits the inverse $\overset{\ast }{g}%
\;^{\!\!-1}=(\overset{\ast }{g}\;^{\!\!jk})$, whose entries are%
\begin{equation}
\overset{\ast }{g}\;^{\!\!jk}=e^{-2\sigma (x)}(2-n\delta
^{jk})G_{1[n]}^{-2/n}y_{1}^{j}y_{1}^{k}\text{ (no sum by }j\text{ or }k\text{%
).}  \label{g-sus-(jk)}
\end{equation}

Let us consider that the Christoffel symbol of the Riemannian metric $%
h_{11}(t)$ on $\mathbb{R}$ is%
\begin{equation*}
\text{\textsc{k}}_{11}^{1}=\frac{h^{11}}{2}\frac{dh_{11}}{dt},
\end{equation*}%
where $h^{11}=1/h_{11}>0$. Then, using a general formula from \cite%
{Balan-Neagu} and taking into account that we have%
\begin{equation*}
\frac{\partial G_{1[n]}}{\partial y_{1}^{i}}=\frac{G_{1[n]}}{y_{1}^{i}},
\end{equation*}%
we find the following geometrical result:

\begin{proposition}
For the $(t,x)$-conformal deformed Berwald-Mo\'{o}r metric (\ref{rheon-B-M}%
), the \textit{energy action functional}%
\begin{equation*}
\overset{\ast }{\mathbf{E}}(t,x(t))=\int_{a}^{b}\overset{\ast }{F^{2}}(t,x,y)%
\sqrt{h_{11}}dt=\int_{a}^{b}e^{2\sigma (x)}\left[
y_{1}^{1}y_{1}^{2}...y_{1}^{n}\right] ^{2/n}\cdot h^{11}\sqrt{h_{11}}dt,
\end{equation*}%
where $y=dx/dt,$ produces on the $1$-jet space $J^{1}(\mathbb{R},M^{n})$ the 
\textbf{canonical nonlinear connection}%
\begin{equation}
\overset{\ast }{\Gamma }=\left( M_{(1)1}^{(i)}=-\text{\textsc{k}}%
_{11}^{1}y_{1}^{i},\text{ }N_{(1)j}^{(i)}=n\sigma _{i}y_{1}^{i}\delta
_{j}^{i}\right) ,  \label{nlc-B-M}
\end{equation}%
where%
\begin{equation*}
\sigma _{i}=\frac{\partial \sigma }{\partial x^{i}}.
\end{equation*}
\end{proposition}

\begin{proof}
For the energy action functional $\overset{\ast }{\mathbf{E}}$, the
associated Euler-Lagrange equations can be written in the equivalent form
(see \cite{Balan-Neagu})%
\begin{equation}
\frac{d^{2}x^{i}}{dt^{2}}+2H_{(1)1}^{(i)}\left( t,x^{k},y_{1}^{k}\right)
+2G_{(1)1}^{(i)}\left( t,x^{k},y_{1}^{k}\right) =0,  \label{Euler-Lagrange=0}
\end{equation}%
where the local components%
\begin{equation*}
H_{(1)1}^{(i)}\overset{def}{=}-\dfrac{1}{2}\text{\textsc{k}}%
_{11}^{1}(t)y_{1}^{i}
\end{equation*}%
and%
\begin{equation*}
\begin{array}{lll}
G_{(1)1}^{(i)} & \overset{def}{=} & \dfrac{h_{11}\overset{\ast }{g}%
\;^{\!\!ip}}{4}\left[ \dfrac{\partial ^{2}\overset{\ast }{F^{2}}}{\partial
x^{r}\partial y_{1}^{p}}y_{1}^{r}-\dfrac{\partial \overset{\ast }{F^{2}}}{%
\partial x^{p}}+\dfrac{\partial ^{2}\overset{\ast }{F^{2}}}{\partial
t\partial y_{1}^{p}}+\right. \medskip \\ 
&  & \left. +\dfrac{\partial \overset{\ast }{F^{2}}}{\partial y_{1}^{p}}%
\text{\textsc{k}}_{11}^{1}(t)+2h^{11}\text{\textsc{k}}_{11}^{1}\overset{\ast 
}{g}_{pr}y_{1}^{r}\right] =\dfrac{n}{2}\sigma _{i}\left( y_{1}^{i}\right)
^{2}%
\end{array}%
\end{equation*}%
represent, from a geometrical point of view, a \textit{spray} on the $1$-jet
space $J^{1}(\mathbb{R},M^{n})$.

Therefore, the \textit{canonical nonlinear connection} associated to this
spray has the local components (see \cite{Balan-Neagu})%
\begin{equation*}
\begin{array}{cc}
M_{(1)1}^{(i)}\overset{def}{=}2H_{(1)1}^{(i)}=-\text{\textsc{k}}%
_{11}^{1}y_{1}^{i}, & N_{(1)j}^{(i)}\overset{def}{=}\dfrac{\partial
G_{(1)1}^{(i)}}{\partial y_{1}^{j}}=n\sigma _{i}y_{1}^{i}\delta _{j}^{i}.%
\end{array}%
\end{equation*}%
\bigskip
\end{proof}

\section{The Cartan canonical $\protect\overset{\ast }{\Gamma }$-linear
connection. Its d-torsions and d-curvatures}

\hspace{5mm}The nonlinear connection (\ref{nlc-B-M}) produces the dual 
\textit{adapted bases} of d-vector fields (no sum by $i$) 
\begin{equation}
\left\{ \frac{\delta }{\delta t}=\frac{\partial }{\partial t}+\text{\textsc{k%
}}_{11}^{1}y_{1}^{p}\frac{\partial }{\partial y_{1}^{p}}\text{ };\text{ }%
\frac{\delta }{\delta x^{i}}=\frac{\partial }{\partial x^{i}}-n\sigma
_{i}y_{1}^{i}\frac{\partial }{\partial y_{1}^{i}}\text{ };\text{ }\dfrac{%
\partial }{\partial y_{1}^{i}}\right\} \subset \mathcal{X}(E)  \label{a-b-v}
\end{equation}%
and d-covector fields (no sum by $i$)%
\begin{equation}
\left\{ dt\text{ };\text{ }dx^{i}\text{ };\text{ }\delta
y_{1}^{i}=dy_{1}^{i}-\text{\textsc{k}}_{11}^{1}y_{1}^{i}dt+n\sigma
_{i}y_{1}^{i}dx^{i}\right\} \subset \mathcal{X}^{\ast }(E),  \label{a-b-co}
\end{equation}%
where $E=J^{1}(\mathbb{R},M^{n})$. The naturalness of the geometrical
adapted bases (\ref{a-b-v}) and (\ref{a-b-co}) is coming from the fact that,
via a transformation of coordinates (\ref{tr-rules}), their elements
transform as the classical tensors. Therefore, the description of all
subsequent geometrical objects on the $1$-jet space $J^{1}(\mathbb{R},M^{n})$
(e.g., the Cartan canonical linear connection, its torsion and curvature
etc.) will be given in local adapted components. Consequently, by direct
computations, we obtain the following geometrical result:

\begin{proposition}
The Cartan canonical $\overset{\ast }{\Gamma }$-linear connection, produced
by the $(t,x)$-conformal deformed Berwald-Mo\'{o}r metric (\ref{rheon-B-M}),
has the following adapted local components (no sum by $i,$ $j$ or $k$):%
\begin{equation}
C\overset{\ast }{\Gamma }=\left( \text{\textsc{k}}_{11}^{1},\text{ }%
G_{j1}^{k}=0,\text{ }L_{jk}^{i}=n\delta _{j}^{i}\delta _{k}^{i}\sigma _{i},%
\text{ }C_{j(k)}^{i(1)}=\mathtt{C}_{jk}^{i}\cdot \frac{y_{1}^{i}}{%
y_{1}^{j}y_{1}^{k}}\right) ,  \label{Cartan-can-connect}
\end{equation}%
where%
\begin{equation*}
\mathtt{C}_{jk}^{i}=-\frac{2}{n^{2}}+\frac{\delta _{j}^{i}+\delta
_{k}^{i}+\delta _{jk}}{n}-\delta _{j}^{i}\delta _{k}^{i}.
\end{equation*}
\end{proposition}

The adapted components of the Cartan canonical connection are given by the
formulas (see \cite{Balan-Neagu})%
\begin{equation*}
\begin{array}{l}
\medskip G_{j1}^{k}\overset{def}{=}\dfrac{\overset{\ast }{g}\;^{\!\!km}}{2}%
\dfrac{\delta \overset{\ast }{g}_{mj}}{\delta t}=0, \\ 
\medskip L_{jk}^{i}\overset{def}{=}\dfrac{\overset{\ast }{g}\;^{\!\!im}}{2}%
\left( \dfrac{\delta \overset{\ast }{g}_{jm}}{\delta x^{k}}+\dfrac{\delta 
\overset{\ast }{g}_{km}}{\delta x^{j}}-\dfrac{\delta \overset{\ast }{g}_{jk}%
}{\delta x^{m}}\right) =n\delta _{j}^{i}\delta _{k}^{i}\sigma _{i}, \\ 
C_{j(k)}^{i(1)}\overset{def}{=}\dfrac{\overset{\ast }{g}\;^{\!\!im}}{2}%
\left( \dfrac{\partial \overset{\ast }{g}_{jm}}{\partial y_{1}^{k}}+\dfrac{%
\partial \overset{\ast }{g}_{km}}{\partial y_{1}^{j}}-\dfrac{\partial 
\overset{\ast }{g}_{jk}}{\partial y_{1}^{m}}\right) =\dfrac{\overset{\ast }{g%
}\;^{\!\!im}}{2}\dfrac{\partial \overset{\ast }{g}_{jm}}{\partial y_{1}^{k}}=%
\mathtt{C}_{jk}^{i}\cdot \dfrac{y_{1}^{i}}{y_{1}^{j}y_{1}^{k}},%
\end{array}%
\end{equation*}%
where we also used the equality%
\begin{equation*}
\frac{\delta \overset{\ast }{g}_{jm}}{\delta x^{k}}=n\delta _{jk}\overset{%
\ast }{g}_{jm}\sigma _{k}+n\delta _{mk}\overset{\ast }{g}_{jm}\sigma _{k}.
\end{equation*}

\begin{remark}
It is important to note that the vertical d-tensor $C_{j(k)}^{i(1)}$ also
has the properties (see also \cite{Mats-Shimada}, \cite{Neagu-x-B-M(3)} and 
\cite{Neagu-x-B-M(4)}):%
\begin{equation}
C_{j(k)}^{i(1)}=C_{k(j)}^{i(1)},\quad C_{j(m)}^{i(1)}y_{1}^{m}=0,\quad
C_{j(m)}^{m(1)}=0,\quad C_{i(k)|m}^{m(1)}=0,  \label{equalitie-C}
\end{equation}%
with sum by $m$, where%
\begin{equation*}
{C_{i(k)|j}^{l(1)}\overset{def}{=}}\frac{\delta {C_{i(k)}^{l(1)}}}{\delta
x^{j}}+{C_{i(k)}^{r(1)}L_{rj}^{l}}-{C_{r(k)}^{l(1)}L_{ij}^{r}}-{%
C_{i(r)}^{l(1)}L_{kj}^{r}.}
\end{equation*}
\end{remark}

\begin{proposition}
The Cartan canonical connection of the $(t,x)$-conformal deformed Berwald-Mo%
\'{o}r metric (\ref{rheon-B-M}) has \textbf{two} effective local torsion
d-tensors:%
\begin{equation*}
\begin{array}{c}
\medskip R_{(1)ij}^{(r)}=n\left( \delta _{i}^{r}\sigma _{rj}-\delta
_{j}^{r}\sigma _{ri}\right) y_{1}^{r}, \\ 
P_{i(j)}^{r(1)}=\left( -\dfrac{2}{n^{2}}+\dfrac{\delta _{i}^{r}+\delta
_{j}^{r}+\delta _{ij}}{n}-\delta _{i}^{r}\delta _{j}^{r}\right) \cdot \dfrac{%
y_{1}^{r}}{y_{1}^{i}y_{1}^{j}},%
\end{array}%
\end{equation*}%
where%
\begin{equation*}
\sigma _{pq}:=\dfrac{\partial ^{2}\sigma }{\partial x^{p}\partial x^{q}}.
\end{equation*}
\end{proposition}

\begin{proof}
Generally, an $h$\textit{-normal }$\Gamma $\textit{-linear connection} on
the $1$-jet space $J^{1}(\mathbb{R},M^{n})$ has \textit{eight} effective
local d-tensors of torsion (for more details, see \cite{Balan-Neagu}). For
the Cartan canonical connection (\ref{Cartan-can-connect}) these reduce only
to \textit{two} (the other six are zero):%
\begin{equation*}
R_{(1)ij}^{(r)}\overset{def}{=}{\dfrac{\delta N_{(1)i}^{(r)}}{\delta x^{j}}}-%
{\dfrac{\delta N_{(1)j}^{(r)}}{\delta x^{i}}},\qquad P_{i(j)}^{r(1)}\overset{%
def}{=}C_{i(j)}^{r(1)}.
\end{equation*}
\end{proof}

\begin{proposition}
The Cartan canonical connection of the $(t,x)$-conformal deformed Berwald-Mo%
\'{o}r metric (\ref{rheon-B-M}) has \textbf{three} effective local curvature
d-tensors:%
\begin{equation*}
\begin{array}{c}
\medskip {R_{ijk}^{l}={\dfrac{\partial L_{ij}^{l}}{\partial x^{k}}}-{\dfrac{%
\partial L_{ik}^{l}}{\partial x^{j}}}%
+L_{ij}^{r}L_{rk}^{l}-L_{ik}^{r}L_{rj}^{l}+C_{i(r)}^{l(1)}R_{(1)jk}^{(r)},%
\qquad}P_{ij(k)}^{l\text{ }(1)}={-C_{i(k)|j}^{l(1)},} \\ 
S_{i(j)(k)}^{l(1)(1)}\overset{def}{=}{{\dfrac{\partial C_{i(j)}^{l(1)}}{%
\partial y_{1}^{k}}}-{\dfrac{\partial C_{i(k)}^{l(1)}}{\partial y_{1}^{j}}}%
+C_{i(j)}^{r(1)}C_{r(k)}^{l(1)}-C_{i(k)}^{r(1)}C_{r(j)}^{l(1)}.}%
\end{array}%
\end{equation*}
\end{proposition}

\begin{proof}
Generally, an $h$\textit{-normal }$\Gamma $\textit{-linear connection} on
the $1$-jet space $J^{1}(\mathbb{R},M^{n})$ has \textit{five} effective
local d-tensors of curvature (for more details, see \cite{Balan-Neagu}). For
the Cartan canonical connection (\ref{Cartan-can-connect}) these reduce only
to \textit{three} (the other two are zero); these are%
\begin{equation*}
\begin{array}{l}
\medskip {R_{ijk}^{l}\overset{def}{=}{\dfrac{\delta L_{ij}^{l}}{\delta x^{k}}%
}-{\dfrac{\delta L_{ik}^{l}}{\delta x^{j}}}%
+L_{ij}^{r}L_{rk}^{l}-L_{ik}^{r}L_{rj}^{l}+C_{i(r)}^{l(1)}R_{(1)jk}^{(r)},}
\\ 
\medskip {P_{ij(k)}^{l\;\;(1)}\overset{def}{=}{\dfrac{\partial L_{ij}^{l}}{%
\partial y_{1}^{k}}}-C_{i(k)|j}^{l(1)}+C_{i(r)}^{l(1)}P_{(1)j(k)}^{(r)\;%
\;(1)}=-C_{i(k)|j}^{l(1)},} \\ 
S_{i(j)(k)}^{l(1)(1)}\overset{def}{=}{{\dfrac{\partial C_{i(j)}^{l(1)}}{%
\partial y_{1}^{k}}}-{\dfrac{\partial C_{i(k)}^{l(1)}}{\partial y_{1}^{j}}}%
+C_{i(j)}^{r(1)}C_{r(k)}^{l(1)}-C_{i(k)}^{r(1)}C_{r(j)}^{l(1)},}%
\end{array}%
\end{equation*}%
where we used the equality%
\begin{equation*}
{P_{(1)j(k)}^{(r)\text{ }(1)}\overset{def}{=}{\dfrac{\partial N_{(1)j}^{(r)}%
}{\partial y_{1}^{k}}}-L_{jk}^{r}=0}.
\end{equation*}
\end{proof}

\section{Field-like geometrical models associated to the $(t,x)$-conformal
deformation of the Berwald-Mo\'{o}r metric}

\subsection{Gravitational-like geometrical model}

\hspace{5mm}The $(t,x)$-conformal deformed Berwald-Mo\'{o}r metric (\ref%
{rheon-B-M}) produces on the $1$-jet space $J^{1}(\mathbb{R},M^{n})$ the
adapted metrical d-tensor (sum by $i$ and $j$)%
\begin{equation}
\mathbf{G}=h_{11}dt\otimes dt+\overset{\ast }{g}_{ij}dx^{i}\otimes
dx^{j}+h^{11}\overset{\ast }{g}_{ij}\delta y_{1}^{i}\otimes \delta y_{1}^{j},
\label{gravit-pot-B-M}
\end{equation}%
where $\overset{\ast }{g}_{ij}$ is given by (\ref{BM-fdt-metr-d-tensor}),
and we have%
\begin{equation*}
\delta y_{1}^{i}=dy_{1}^{i}-\text{\textsc{k}}_{11}^{1}y_{1}^{i}dt+n\sigma
_{i}y_{1}^{i}dx^{i}\text{ (no sum by }i\text{).}
\end{equation*}

From an abstract physical point of view, the metrical d-tensor (\ref%
{gravit-pot-B-M}) may be regarded as a \textit{\textquotedblleft
non-isotropic gravitational potential\textquotedblright }. In our
geometric-physical approach, one postulates that the non-isotropic
gravitational potential $\mathbf{G}$ is governed by the following \textit{%
geometrical Einstein-like equations:}%
\begin{equation}
\text{Ric }(C\overset{\ast }{\Gamma })-\frac{\text{Sc }(C\overset{\ast }{%
\Gamma })}{2}\mathbf{G}\mathbb{=}\mathcal{KT},  \label{Einstein-eq-global}
\end{equation}%
where

\begin{itemize}
\item[$\blacklozenge $] Ric $(C\overset{\ast }{\Gamma })$ is the \textit{%
Ricci d-tensor} associated to the Cartan canonical linear connection (\ref%
{Cartan-can-connect});

\item[$\blacklozenge $] Sc $(C\overset{\ast }{\Gamma })$ is the \textit{%
scalar curvature};

\item[$\blacklozenge $] $\mathcal{K}$ is the \textit{Einstein constant} and $%
\mathcal{T}$ is the intrinsic \textit{non-isotropic stress-energy d-tensor
of matter}.
\end{itemize}

Therefore, using the adapted basis of vector fields (\ref{a-b-v}), we can
locally describe the global geometrical Einstein-like equations (\ref%
{Einstein-eq-global}). Consequently, some direct computations lead to:

\begin{lemma}
The Ricci tensor of the Cartan canonical connection $C\overset{\ast }{\Gamma 
}$ of the $(t,x)$-conformal deformed Berwald-Mo\'{o}r metric (\ref{rheon-B-M}%
) has the following \textbf{two} effective local Ricci d-tensors (no sum by $%
i$, $j$, $k$ or $l$):%
\begin{equation}
\begin{array}{l}
\bigskip R_{ij}=\left\{ 
\begin{array}{ll}
-\sigma _{ij}-\underset{\overset{m=1}{m\neq j}}{\overset{n}{\sum }}\sigma
_{jm}\dfrac{y_{1}^{m}}{y_{1}^{i}}, & i\neq j\medskip \\ 
0, & i=j,%
\end{array}%
\right. \\ 
S_{(i)(j)}^{(1)(1)}=\left[ \dfrac{2}{n^{2}}-\dfrac{1}{n}+\left( 1-\dfrac{2}{n%
}\right) \delta _{ij}\right] \cdot \dfrac{1}{y_{1}^{i}y_{1}^{j}}{.}%
\end{array}
\label{Ricci-local}
\end{equation}
\end{lemma}

\begin{proof}
Generally, the Ricci tensor of the Cartan canonical connection $C\Gamma $
produced by an arbitrary jet Lagrangian function is determined by \textit{six%
} effective local Ricci d-tensors (for more details, see \cite{Balan-Neagu}%
). For our particular Cartan canonical connection (\ref{Cartan-can-connect})
these reduce only to the following \textit{two} (the other four are zero):%
\begin{equation*}
\begin{array}{lll}
\bigskip R_{ij} & \overset{def}{=} & R_{ijm}^{m}={{\dfrac{\partial L_{ij}^{m}%
}{\partial x^{m}}}-{\dfrac{\partial L_{im}^{m}}{\partial x^{j}}}%
+L_{ij}^{r}L_{rm}^{m}-L_{im}^{r}L_{rj}^{m}+C_{i(r)}^{m(1)}R_{(1)jm}^{(r)},}
\\ 
\bigskip S_{(i)(j)}^{(1)(1)} & \overset{def}{=} & S_{i(j)(m)}^{m(1)(1)}={{%
\dfrac{\partial C_{i(j)}^{m(1)}}{\partial y_{1}^{m}}}-{\dfrac{\partial
C_{i(m)}^{m(1)}}{\partial y_{1}^{j}}}%
+C_{i(j)}^{r(1)}C_{r(m)}^{m(1)}-C_{i(m)}^{r(1)}C_{r(j)}^{m(1)}=} \\ 
& {=} & {{\dfrac{\partial C_{i(j)}^{m(1)}}{\partial y_{1}^{m}}}%
-C_{i(m)}^{r(1)}C_{r(j)}^{m(1)},}%
\end{array}%
\end{equation*}%
with sum by $r$ and $m$.
\end{proof}

\begin{lemma}
The scalar curvature of the Cartan canonical connection $C\overset{\ast }{%
\Gamma }$ of the $(t,x)$-conformal deformed Berwald-Mo\'{o}r metric (\ref%
{rheon-B-M}) has the value%
\begin{equation*}
\text{\emph{Sc} }(C\overset{\ast }{\Gamma })=-e^{-2\sigma }G_{1[n]}^{-2/n}%
\left[ 4nY_{11}+\left( n^{2}-3n+2\right) h_{11}\right] ,
\end{equation*}%
where%
\begin{equation*}
Y_{11}=\underset{\overset{p,q=1}{p<q}}{\overset{n}{\sum }}\sigma
_{pq}y_{1}^{p}y_{1}^{q}.
\end{equation*}
\end{lemma}

\begin{proof}
The scalar curvature of the Cartan canonical connection (\ref%
{Cartan-can-connect}) is given by the formula (for more details, see \cite%
{Balan-Neagu})%
\begin{equation*}
\text{Sc }(C\overset{\ast }{\Gamma })=\overset{\ast }{g}\;^{\!%
\!pq}R_{pq}+h_{11}\overset{\ast }{g}\;^{\!\!pq}S_{(p)(q)}^{(1)(1)}.
\end{equation*}
\end{proof}

The local description in the adapted basis of vector fields (\ref{a-b-v}) of
the global geometrical Einstein-like equations (\ref{Einstein-eq-global}) is
given by (for more details, see \cite{Balan-Neagu}):

\begin{proposition}
The \textbf{geometrical Einstein-like equations} produced by the $(t,x)$%
-conformal deformed Berwald-Mo\'{o}r metric (\ref{rheon-B-M}) are locally
described by:%
\begin{equation}
\left\{ 
\begin{array}{l}
\medskip e^{-2\sigma }G_{1[n]}^{-2/n}\left[ 2nY_{11}+\dfrac{n^{2}-3n+2}{2}%
h_{11}\right] h_{11}=\mathcal{KT}_{11} \\ 
\medskip R_{ij}+e^{-2\sigma }G_{1[n]}^{-2/n}\left[ 2nY_{11}+\dfrac{n^{2}-3n+2%
}{2}h_{11}\right] \overset{\ast }{g}_{ij}=\mathcal{KT}_{ij} \\ 
S_{(i)(j)}^{(1)(1)}+e^{-2\sigma }G_{1[n]}^{-2/n}\left[ 2nY_{11}+\dfrac{%
n^{2}-3n+2}{2}h_{11}\right] h^{11}\overset{\ast }{g}_{ij}=\mathcal{KT}%
_{(i)(j)}^{(1)(1)}\medskip \\ 
\begin{array}{lll}
0=\mathcal{T}_{1i}, & 0=\mathcal{T}_{i1}, & 0=\mathcal{T}_{(i)1}^{(1)}%
\medskip \\ 
0=\mathcal{T}_{1(i)}^{\text{ }(1)}, & 0=\mathcal{T}_{i(j)}^{\text{ }(1)}, & 
0=\mathcal{T}_{(i)j}^{(1)}.%
\end{array}%
\end{array}%
\right.  \label{E-1}
\end{equation}
\end{proposition}

\begin{corollary}
The non-isotropic stress-energy d-tensor of matter $\mathcal{T}$ satisfies
the following \textbf{geometrical conservation laws} (sum by $m$):%
\begin{equation*}
\left\{ 
\begin{array}{l}
\bigskip \mathcal{T}_{1/1}^{1}+\mathcal{T}_{1|m}^{m}+\mathcal{T}%
_{(1)1}^{(m)}|_{(m)}^{(1)}=0 \\ 
\mathcal{T}_{i/1}^{1}+\mathcal{T}_{i|m}^{m}+\mathcal{T}%
_{(1)i}^{(m)}|_{(m)}^{(1)}=E_{i|m}^{m}\medskip \\ 
\mathcal{T}_{\text{ \ }(i)/1}^{1(1)}+\mathcal{T}_{\text{ \ }(i)|m}^{m(1)}+%
\mathcal{T}_{(1)(i)}^{(m)(1)}|_{(m)}^{(1)}=\dfrac{2e^{-2\sigma
}G_{1[n]}^{-2/n}}{\mathcal{K}}\cdot \left[ n\dfrac{\partial Y_{11}}{\partial
y_{1}^{i}}-2\dfrac{Y_{11}}{y_{1}^{i}}\right] ,%
\end{array}%
\right.
\end{equation*}%
where (sum by $r$):$\medskip $

$\bigskip \mathcal{T}_{1}^{1}\overset{def}{=}h^{11}\mathcal{T}_{11}=\mathcal{%
K}^{-1}e^{-2\sigma }G_{1[n]}^{-2/n}\left[ 2nY_{11}+\dfrac{n^{2}-3n+2}{2}%
h_{11}\right] ,$

$\bigskip \mathcal{T}_{1}^{m}\overset{def}{=}\overset{\ast }{g}\;^{\!\!mr}%
\mathcal{T}_{r1}=0,\quad \mathcal{T}_{(1)1}^{(m)}\overset{def}{=}h_{11}%
\overset{\ast }{g}\;^{\!\!mr}\mathcal{T}_{(r)1}^{(1)}=0,\quad \mathcal{T}%
_{i}^{1}\overset{def}{=}h^{11}\mathcal{T}_{1i}=0,$

$\bigskip \mathcal{T}_{i}^{m}\overset{def}{=}\overset{\ast }{g}\;^{\!\!mr}%
\mathcal{T}_{ri}:=E_{i}^{m}=\mathcal{K}^{-1}\left[ \dfrac{{}}{{}}\overset{%
\ast }{g}\;^{\!\!mr}R_{ri}+\right. $

$\bigskip \qquad\left. +e^{-2\sigma }G_{1[n]}^{-2/n}\left( 2nY_{11}+\dfrac{%
n^{2}-3n+2}{2}h_{11}\right) \delta _{i}^{m}\right] ,$

$\bigskip \mathcal{T}_{(1)i}^{(m)}\overset{def}{=}h_{11}\overset{\ast }{g}%
\;^{\!\!mr}\mathcal{T}_{(r)i}^{(1)}=0,\;\;\mathcal{T}_{\text{ \ }(i)}^{1(1)}%
\overset{def}{=}h^{11}\mathcal{T}_{1(i)}^{\text{ }(1)}=0,\;\;\mathcal{T}_{%
\text{ \ }(i)}^{m(1)}\overset{def}{=}\overset{\ast }{g}\;^{\!\!mr}\mathcal{T}%
_{r(i)}^{\text{ }(1)}=0,$

$\medskip \mathcal{T}_{(1)(i)}^{(m)(1)}\overset{def}{=}h_{11}\overset{\ast }{%
g}\;^{\!\!mr}\mathcal{T}_{(r)(i)}^{(1)(1)}=\dfrac{e^{-2\sigma
}G_{1[n]}^{-2/n}}{\mathcal{K}}\cdot \left[ \dfrac{n-2}{n}h_{11}\dfrac{%
y_{1}^{m}}{y_{1}^{i}}+\right. $

$\bigskip \qquad\left. +\left( 2nY_{11}+\dfrac{n^{2}-5n+6}{2}h_{11}\right)
\delta _{i}^{m}\right] ,$

and we also have (summation by $m$ and $r$, but no sum by $i$)$\bigskip $

$\bigskip\mathcal{T}_{1/1}^{1}=\dfrac{\delta \mathcal{T}_{1}^{1}}{\delta t}%
,\quad\mathcal{T}_{1|m}^{m}\overset{def}{=}\dfrac{\delta \mathcal{T}_{1}^{m}%
}{\delta x^{m}}+\mathcal{T}_{1}^{r}L_{rm}^{m},$

$\bigskip \mathcal{T}_{(1)1}^{(m)}|_{(m)}^{(1)}\overset{def}{=}\dfrac{%
\partial \mathcal{T}_{(1)1}^{(m)}}{\partial y_{1}^{m}}+\mathcal{T}%
_{(1)1}^{(r)}C_{r(m)}^{m(1)}=\dfrac{\partial \mathcal{T}_{(1)1}^{(m)}}{%
\partial y_{1}^{m}},$

$\bigskip \mathcal{T}_{i/1}^{1}\overset{def}{=}\dfrac{\delta \mathcal{T}%
_{i}^{1}}{\delta t}+\mathcal{T}_{i}^{1}\text{\textsc{k}}_{11}^{1}-\mathcal{T}%
_{r}^{1}G_{i1}^{r}=\dfrac{\delta \mathcal{T}_{i}^{1}}{\delta t}+\mathcal{T}%
_{i}^{1}\text{\textsc{k}}_{11}^{1},$

$\bigskip \mathcal{T}_{i|m}^{m}\overset{def}{=}\dfrac{\delta \mathcal{T}%
_{i}^{m}}{\delta x^{m}}+\mathcal{T}_{i}^{r}L_{rm}^{m}-\mathcal{T}%
_{r}^{m}L_{im}^{r}=E_{i|m}^{m}:=\dfrac{\delta E_{i}^{m}}{\delta x^{m}}%
+nE_{i}^{m}\sigma _{m}-nE_{i}^{i}\sigma _{i},$

$\bigskip \mathcal{T}_{(1)i}^{(m)}|_{(m)}^{(1)}\overset{def}{=}\dfrac{%
\partial \mathcal{T}_{(1)i}^{(m)}}{\partial y_{1}^{m}}+\mathcal{T}%
_{(1)i}^{(r)}C_{r(m)}^{m(1)}-\mathcal{T}_{(1)r}^{(m)}C_{i(m)}^{r(1)}=\dfrac{%
\partial \mathcal{T}_{(1)i}^{(m)}}{\partial y_{1}^{m}}-\mathcal{T}%
_{(1)r}^{(m)}C_{i(m)}^{r(1)},$

$\bigskip \mathcal{T}_{\text{ \ }(i)/1}^{1(1)}\overset{def}{=}\dfrac{\delta 
\mathcal{T}_{\text{ \ }(i)}^{1(1)}}{\delta t}+2\mathcal{T}_{\text{ \ }%
(i)}^{1(1)}\text{\textsc{k}}_{11}^{1},\quad\mathcal{T}_{\text{ \ }%
(i)|m}^{m(1)}\overset{def}{=}\dfrac{\delta \mathcal{T}_{\text{ \ }(i)}^{m(1)}%
}{\delta x^{m}}+\mathcal{T}_{\text{ \ }(i)}^{r(1)}L_{rm}^{m}-\mathcal{T}_{%
\text{ \ }(r)}^{m(1)}L_{im}^{r},$

$\mathcal{T}_{(1)(i)}^{(m)(1)}|_{(m)}^{(1)}\overset{def}{=}\dfrac{\partial 
\mathcal{T}_{(1)(i)}^{(m)(1)}}{\partial y_{1}^{m}}+\mathcal{T}%
_{(1)(i)}^{(r)(1)}C_{r(m)}^{m(1)}-\mathcal{T}%
_{(1)(r)}^{(m)(1)}C_{i(m)}^{r(1)}=\dfrac{\partial \mathcal{T}%
_{(1)(i)}^{(m)(1)}}{\partial y_{1}^{m}}.\medskip $
\end{corollary}

\begin{proof}
The local Einstein-like equations (\ref{E-1}), together with some direct
computations, lead us to what we were looking for. Also note that we have
(summation by $m$ and $r$)%
\begin{equation*}
\mathcal{T}_{(1)(r)}^{(m)(1)}C_{i(m)}^{r(1)}=0.
\end{equation*}
\end{proof}

\subsection{Electromagnetic-like geometrical model}

\hspace{5mm}In the book \cite{Balan-Neagu}, a geometrical theory for
electromagnetism was also created, using only a given Lagrangian function $L$
on the $1$-jet space $J^{1}(\mathbb{R},M^{n})$. In the background of the jet
single-time (one-parameter) Lagrange geometry from \cite{Balan-Neagu}, one
works with the following \textit{non-isotropic electromagnetic distinguished 
}$2$\textit{-form} (sum by $i$ and $j$):%
\begin{equation*}
\mathbf{F}=F_{(i)j}^{(1)}\delta y_{1}^{i}\wedge dx^{j},
\end{equation*}%
where (sum by $m$ and $r$)%
\begin{equation*}
F_{(i)j}^{(1)}=\frac{h^{11}}{2}\left[ \overset{\ast }{g}_{jm}N_{(1)i}^{(m)}-%
\overset{\ast }{g}_{im}N_{(1)j}^{(m)}+\left( \overset{\ast }{g}%
_{ir}L_{jm}^{r}-\overset{\ast }{g}_{jr}L_{im}^{r}\right) y_{1}^{m}\right] .
\end{equation*}%
This is characterized by some natural \textit{geometrical Maxwell-like
equations} (for more details, see \cite{Mir-An} and \cite{Balan-Neagu}).

\begin{remark}
The Lagrangian function that governs the movement law of a particle of mass $%
m\neq 0$ and electric charge $e$, which is displaced concomitantly into an
environment endowed both with a gravitational field and an electromagnetic
one, is given by 
\begin{equation}
L(t,x^{k},y_{1}^{k})=mch^{11}(t)\;\varphi _{ij}(x^{k})\;y_{1}^{i}y_{1}^{j}+{%
\frac{2e}{m}}A_{(i)}^{(1)}(t,x^{k})\;y_{1}^{i},  \label{x05-L-(jet)-ED}
\end{equation}%
where

\begin{itemize}
\item the semi-Riemannian metric $\varphi _{ij}(x)$ represents the \textbf{%
isotropic gravitational potential};

\item $A_{(i)}^{(1)}(t,x)$ are the components of a d-tensor on the $1$-jet
space $J^{1}(\mathbb{R},M^{n})$ representing the \textbf{electromagnetic
potential}.
\end{itemize}

Note that the jet Lagrangian function (\ref{x05-L-(jet)-ED}) is a natural
extension of the Lagrangian (defined on the tangent bundle) used in
electrodynamics by Miron and Anastasiei \cite{Mir-An}. In our jet Lagrangian
formalism applied to (\ref{x05-L-(jet)-ED}), the \textbf{%
electromagnetic-like components} become classical ones (see \cite%
{Balan-Neagu}):%
\begin{equation*}
F_{(i)j}^{(1)}=-\frac{e}{2m}\left( \frac{\partial A_{(i)}^{(1)}}{\partial
x^{j}}-\frac{\partial A_{(j)}^{(1)}}{\partial x^{i}}\right) .
\end{equation*}%
Moreover, the second set of \textbf{geometrical Maxwell-like equations}\
reduce to the classical ones too (for more details, see \cite{Mir-An}, \cite%
{Balan-Neagu}):%
\begin{equation*}
\sum_{\{i,j,k\}}F_{(i)j|k}^{(1)}=0,
\end{equation*}%
where 
\begin{equation*}
F_{(i)j|k}^{(1)}=\frac{\partial F_{(i)j}^{(1)}}{\partial x^{k}}%
-F_{(m)j}^{(1)}\gamma _{ik}^{m}-F_{(i)m}^{(1)}\gamma _{jk}^{m}.
\end{equation*}%
Also, the \textbf{geometrical Einstein-like equations} attached to the
Lagrangian (\ref{x05-L-(jet)-ED}) (see \cite{Mir-An}, \cite{Balan-Neagu})
are the same with the famous classical ones (associated to the
semi-Riemannian metric $\varphi _{ij}(x)$). In author's opinion, these facts
suggest some kind of naturalness for the present abstract Lagrangian
non-isotropic electromagnetic and gravitational geometrical theory.
\end{remark}

Via some direct calculations, we easily deduce that the $(t,x)$-conformal
deformed Berwald-Mo\'{o}r metric (\ref{rheon-B-M}) produces null
non-isotropic electromagnetic components:%
\begin{equation*}
F_{(i)j}^{(1)}=0.
\end{equation*}%
It follows that our $(t,x)$-conformal deformed jet Berwald-Mo\'{o}r
geometrical electromagnetic-like theory is trivial. This fact probably
suggests that the $(t,x)$-conformal deformed Berwald-Mo\'{o}r geometrical
structure (\ref{rheon-B-M}) has rather gravitational connotations than
electromagnetic ones.

As a conclusion, it is possible for the recent Voicu-Siparov approach of the
electromagnetism in spaces with anisotropic metrics (that electromagnetic
approach is different from the electromagnetic theory exposed above, and it
is developed in the paper \cite{Voicu-Siparov}) to give other interesting
electromagnetic-geometrical results for spaces endowed with the Berwald-Mo%
\'{o}r geometrical structure.

\textbf{Open problem.} The author of this paper believes that the finding of
some possible real physical interpretations for the present non-isotropic
Berwald-Mo\'{o}r geometrical approach of gravity and electromagnetism may be
an open problem for physicists.

Mircea Neagu

University Transilvania of Bra\c{s}ov,

Department of Mathematics - Informatics

Blvd. Iuliu Maniu, no. 50, Bra\c{s}ov 500091, Romania.

\textit{E-mail:} mircea.neagu@unitbv.ro

\end{document}